\documentclass[11pt]{article}
\pdfoutput=1
\usepackage{amsmath,amsthm}
\usepackage{amssymb,latexsym}
\usepackage{graphicx}
\usepackage{mathpazo}
\usepackage{fullpage}
\usepackage{hyperref}
\usepackage{amsfonts}
\usepackage{enumerate}
\usepackage{ifpdf}
\usepackage{hyperref}

\def \RR {\mathbb R}

\def \EE {\mathbb E}

\DeclareMathOperator*{\PP}{\mathbb P}

\def \vphi {\varphi}

\def \cI {\mathcal I}

\def \cS {\mathcal S}

\def \cG {\mathcal G}
\def \cV {\mathcal V}

\def \d  {{\rm d}}

\def \poly {{\rm poly}}

\def \VSP {{\textnormal{\textsf{VSP}}}}

\newtheorem{theorem}{Theorem}[section]
\newtheorem{lemma}[theorem]{Lemma}

\newtheorem{corollary}[theorem]{Corollary}
\theoremstyle{definition}
\newtheorem{definition}[theorem]{Definition}

\newcommand{\myparagraph}[1]{\paragraph{#1.}}
\newcommand{\floor}[1]{{\lfloor #1 \rfloor}}

\newif\ifnotes\notesfalse

\ifnotes
    \newcommand{\bnote}[1]{{\bf (Bo'az:} {#1}{\bf ) }}
    \newcommand{\onote}[1]{{\bf (Oded:} {#1}{\bf ) }}
    \newcommand{\noteswarning}{{\begin{center} {\Large WARNING: NOTES ON}\end{center}}}
\else
    \newcommand{\bnote}[1]{}
    \newcommand{\onote}[1]{}
    \newcommand{\noteswarning}{{}}
\fi

\begin{document}

\title{\bf Quantum One-Way Communication is Exponentially Stronger Than Classical Communication}

\author{
 Bo'az Klartag\thanks{School of Mathematical Sciences, Tel Aviv University, Tel Aviv 69978, Israel. Supported in part by the Israel Science Foundation and by a Marie Curie Reintegration Grant from the
Commission of the European Communities.} \and
Oded Regev\thanks{Blavatnik School of Computer Science, Tel Aviv University, Tel Aviv 69978, Israel. Supported
   by the Israel Science Foundation,
   by the Wolfson Family Charitable Trust, and by a European Research Council (ERC) Starting Grant.
   Part of the work done while a DIGITEO visitor in LRI, Orsay.
}
}
\date{}

\maketitle

\begin{abstract}
In STOC 1999, Raz presented a (partial) function for which there is a quantum protocol
communicating only $O(\log n)$ qubits, but for which any classical (randomized, bounded-error) protocol requires
$\poly(n)$ bits of communication. That quantum protocol requires two
rounds of communication. Ever since Raz's paper it was open whether the same exponential
separation can be achieved with a quantum protocol that uses only one round of
communication. Here we settle this question in the
affirmative.
\end{abstract}

\noteswarning

\section{Introduction}

Communication complexity is one of the most basic models in computational complexity,
with wide-ranging applications in computer science~\cite{kushilevitznisan-book}.
The typical question asked in this model is the following. Two remote players, call them Alice and Bob,
are each given an input and are trying to compute some function of their inputs while using as little
communication as possible. How much communication is needed in order to compute the function?
The answer to this question often depends on what exactly we mean by ``compute using as little
communication as possible." One of the central models in this area is that of \emph{randomized (bounded-error)
communication}. Here we allow the players to toss coins, and require them to output the correct answer
with probability at least (say) $2/3$ on \emph{any} given input. This model is quite powerful and
corresponds quite well to what is actually achievable in real-world communication. For instance, one of the
most basic results in this area shows that the players can decide if their inputs are equal
using only $O(\log n)$ bits of communication, where $n$ is the size of their inputs in bits.
Another well-established model of communication is that of \emph{quantum communication}~\cite{Yao}.
Here, we allow the players to communicate quantum states, and to perform quantum operations on them.
Although not nearly as common as classical (i.e., non-quantum) communication, this model is able to provide important
insights into the power of quantum mechanics.

The focus of our work is on the relative power of these two central models,
a question whose study started in the late 1990s~\cite{BuhrmanCW98,AmbainisSTVW03}.
Most notably, Raz~\cite{Raz99} presented a (partial) function for which there is a quantum
protocol communicating only $O(\log n)$ qubits, but for which any classical (randomized, bounded-error) protocol
requires $\poly(n)$ bits of communication (i.e., $\Omega(n^c)$ for some constant $c>0$). This result
demonstrates that quantum communication is exponentially stronger than classical communication,
and is one of the most fundamental results in quantum communication complexity.

However, although Raz's function can be computed using only $O(\log n)$ qubits,
it seems to require at least two rounds of communication between Alice
and Bob. This naturally leads to the following fundamental question,
which has been open ever since Raz's paper:
can a similar exponential separation be achieved with a quantum protocol that uses \emph{only one round
of communication}? In other words:
\begin{quote}
\emph{Can quantum one-way communication be exponentially
stronger than classical two-way communication in computing a function}?
\end{quote}
Such a result might be the strongest possible separation between
quantum communication and classical communication.

There have been quite a few partial results in this direction. First, Bar-Yossef,
Jayram, and Kerenidis~\cite{BarYossefJK04} presented a \emph{relational problem} (i.e., one
in which there is possibly more than one correct answer for a given input) that has
a quantum one-way protocol using only $O(\log n)$ qubits of communication,
but for which any classical protocol \emph{using only one round of communication} must
communicate $\poly(n)$ bits. Classical two-way protocols, however, can easily solve their problem
using $O(\log n)$ bits. Their result was improved by Gavinsky, Kempe, Kerenidis,
Raz, and de Wolf~\cite{GavinskyKKRW07} who proved the same separation, namely,
$O(\log n)$ qubit protocol versus a $\poly(n)$ lower bound for any classical one-way protocol,
but in the standard setting of a \emph{functional problem}.
Again, classical two-way protocols can easily solve the problem using only $O(\log n)$ bits.
See also~\cite{Montanato10} for a similar separation.
Another closely related result is by Gavinsky~\cite{Gavinsky08}, who improved on Bar-Yossef et al.'s
\cite{BarYossefJK04} result in the other direction: namely, he showed an exponential separation between one-way quantum communication
and \emph{two-way classical communication} (just as in the open question) but for a \emph{relational problem}. Gavinsky's proof
is quite involved, and it is not clear if his techniques can be used
to attack the functional case.

It is important to note that there is a big difference between relational separations
and functional ones, with the latter often being more interesting, involving
deeper ideas, and having more profound implications.
Indeed, the functional separation in~\cite{GavinskyKKRW07} required the use of a
hypercontractive inequality and also provided a surprising counterexample
to a conjecture regarding extractors that are secure against quantum adversaries.
Moreover, the existence of a relational separation often says little about
the existence of a functional one; for instance, there are cases where relational separations
provably have no functional counterpart~\cite{GavinskyRW08}.

Here we settle the open question by exhibiting a (partial) function for
which there exists a quantum one-way communication protocol using only $O(\log n)$ qubits, but
for which any classical two-way communication protocol must communicate at least $\poly(n)$ bits. The function we
consider is actually the complete problem for one-way quantum communication~\cite{Kremer95}
and was also described in~\cite{Raz99}.
We call it the \emph{Vector in Subspace Problem} ($\VSP$).
In this problem, Alice is given an $n$-dimensional unit vector $u \in S^{n-1}$
and Bob is given a subspace $H \subset \RR^n$ of dimension $n/2$ with the promise that either
$u \in H$ or $u \in H^\perp$. Their goal is to decide which is the case. (For a formal definition
see Section~\ref{sec:comm}.) The quantum protocol for the problem is almost immediate
from the definition: Alice encodes the vector $u$ as a quantum state of $\lceil \log_2 n \rceil$ qubits (by definition,
the state of a quantum system with $k$ qubits is a $2^k$-dimensional unit vector) and sends it to Bob, who,
after having received the quantum state, performs the projective measurement given by $(H,H^\perp)$.
If $u \in H$, Bob is guaranteed to obtain the former outcome; if $u \in H^\perp$, Bob is guaranteed
to obtain the latter outcome.

It is easy to see that $\VSP$ has a classical protocol using $O(n \log n)$ bits: Alice
simply sends the vector $u$ to Bob, by specifying each coordinate to within an additive $\pm1/\poly(n)$ accuracy.
As noted by Raz~\cite{Raz99}, this protocol is not optimal, and the problem actually has an $O(\sqrt{n})$ protocol,
which we will describe in Section~\ref{sec:comm}.

But of course, our focus in this paper is on \emph{lower bounds}.
Our main result is an $\Omega(n^{1/3})$ lower bound on the (classical) communication
complexity of $\VSP$. Previously no lower bound better than logarithmic was known. Our proof involves some techniques that seem novel
in the computer science literature. We use a hypercontractive
inequality, applied in a fashion similar to that in Kahn,
Kalai, and Linial~\cite{KahnKL88} and in other more recent papers, including the result
by Gavinsky et al.\ mentioned above~\cite{GavinskyKKRW07} (see also~\cite{deWolfSurvey}).
However, unlike previous work, our hypercontractive inequality is in the setting of
functions defined on the sphere. We also use the Radon transform and some of its
basic properties, as well as a rather delicate martingale argument.
Finally, we feel that the proof, at least at a very high level,
is conceptually simpler than some of the previous proofs in this
line of work. We hope that our result and techniques will
find other applications.

One obvious open question left by our work
is to improve the lower bound to a tight $\Omega(n^{1/2})$;
we will mention one possible approach below.
Another open question is to strengthen our result
by showing a separation between the quantum simultaneous message passing (SMP) model
and the classical two-way model. This question was
recently answered by Gavinsky~\cite{Gavinsky09} for relational problems,
but the question for functions seems quite challenging, and it is not
even clear if such a separation can exist.
A final important open question is to understand the power of quantum communication
in computing \emph{total} functions; so far the best known separation is
polynomial.

\section{Proof Sketch}

\ifpdf
\begin{figure}
\begin{centering}
\includegraphics[width=0.3\textwidth]{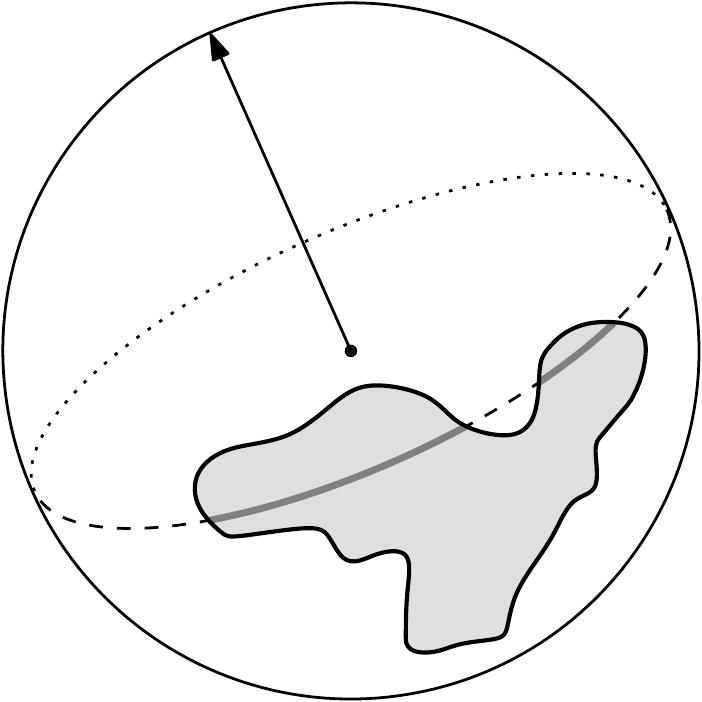}
\caption{A subset of $S^2$ and an equator}
\label{fig:equator}
\end{centering}
\end{figure}
\fi

Here we give an informal sketch of the main ideas in the proof of our lower bound, and
include some remarks regarding the tightness and other aspects of our proofs.
The proof starts in Section~\ref{sec:comm} with a more or less standard application of the rectangle bound
which we do not describe here. This shows that in order to prove our communication lower bound,
it suffices to prove the following sampling statement, which is our main technical theorem (see Figure~\ref{fig:equator} for an illustration).
The formal statement will appear as Theorem~\ref{thm:mainsamplinglowdim}.

\begin{theorem}[Informal]\label{thm:informalsampling}
Let $A$ be an arbitrary (measurable) subset of the sphere $S^{n-1}$ whose measure $\sigma(A)$ (under the uniform probability measure on $S^{n-1}$) is at least $\exp(-n^{1/3})$.
Assume we choose a uniformly random subspace $H \subset \RR^n$ of dimension $n/2$.
Consider the measure of the set $A \cap H$ under the uniform probability measure on the unit sphere $H \cap S^{n-1}$
of the $n/2$-dimensional subspace $H$. Then this measure is within a factor of (say) $1 \pm 0.1$
of $\sigma(A)$ except with probability at most $\exp(-n^{1/3})$.
\end{theorem}

Before we proceed to discuss the proof of this theorem, we make two remarks.
First, it is interesting to note that this theorem is a considerable strengthening of
Lemma 4.1 in~\cite{Raz99}, which is a similar sampling statement, but one that
applies only to sets $A$ whose measure is constant (or slightly less).
Raz proves that lemma using an elementary (but clever) use of Chernoff's
concentration bound. See also the paper by Milman and Wagner~\cite{MilmanW03} for a further
discussion and applications of Raz's sampling lemma.

The second remark is that our theorem is tight in the sense that there exists a set $A$
of measure $\exp(-n^{1/3})$ such that the probability of the measure of $A \cap H$ deviating
by more than $10 \%$ is essentially $\exp(-n^{1/3})$. This set $A$ is simply
a spherical cap, and the bad $H$'s are those that are close to the center of the cap. We omit
this standard calculation. One implication of this is that improving our $\Omega(n^{1/3})$
lower bound to a tight $\Omega(n^{1/2})$ is probably impossible using the rectangle bound,
and one might have to use instead the \emph{smooth rectangle bound} introduced
in~\cite{klauck10,jaink10} and used recently in~\cite{ChakrabartiR10}.
For the interested reader, we note that the following reasonable sampling statement would imply the tight $\Omega(n^{1/2})$ bound.
Let $A$ be an arbitrary subset
of the sphere $S^{n-1}$ whose measure $\sigma(A)$ is at least $\exp(-n^{1/2})$,
and assume we choose a uniformly random subspace $H \subset \RR^n$ of dimension $n/2$.
We now consider the measure of the set $A \cap H$ and that of the set $A \cap H^\perp$
(under the appropriate uniform probability measures). Then the goal would be to prove that
the \emph{average} of these two measures is at least $0.9\,\sigma(A)$ except with probability at most $\exp(-n^{1/2})$.

Theorem~\ref{thm:informalsampling} is proven by a recursive application of the following core sampling statement
for $(n-1)$-dimensional subspaces. Roughly speaking, it shows that sampling a set of measure at least $\exp(-n^{1/3})$
using a random $(n-1)$-dimensional subspace gives an error that is typically at most $1 \pm n^{-2/3}$ and
has an exponential decay.
The formal statement will appear as Theorem~\ref{thm:mainequator}.

\begin{theorem}[Informal]\label{thm:informalsamplingequator}
Let $A \subset S^{n-1}$ be of measure at least $\exp(-n^{1/3})$. Assume we choose a uniformly random subspace $H \subset \RR^n$
of dimension $n-1$. Then, for any $0<t<1$, the measure of $A \cap H$ (under the uniform measure on $H \cap S^{n-1}$)
is within a factor of $1 \pm t$ of $\sigma(A)$ except with probability at most $\exp(-n^{2/3}t)$.
\end{theorem}

Section~\ref{sec:maintechnical} will be dedicated to deriving Theorem~\ref{thm:informalsampling}
from the above theorem. This is done using a martingale argument and Bernstein-type
inequalities; in the following we just give the rough idea. Consider the following equivalent
way to choose a uniformly random subspace $H$ of dimension $n/2$. First, let $H_0=\RR^n$. Then, choose a uniformly random subspace
$H_1 \subset H_0 = \RR^n$ of dimension $n-1$; then, choose a uniformly random subspace $H_2$ of $H_1$
of dimension $n-2$; continue in the same fashion until $H=H_{n/2}$ which is a uniformly random
$n/2$-dimensional subspace of $H_{n/2-1}$. We now consider the sequence of measures
of $A \cap H_i$ (with respect to the uniform measure in $S^{n-1} \cap H_{i}$) for $i=0,\ldots,n/2$.
By definition, this sequence starts with $\sigma(A)$. According to Theorem~\ref{thm:informalsamplingequator}, at each step of the sequence we typically get an extra multiplicative error of
$1 \pm n^{-2/3}$. After $n/2$ steps, the accumulated error becomes $1 \pm \sqrt{n} \cdot n^{-2/3} = 1 \pm n^{-1/6}$
(this of course requires a proof since, e.g., the steps are not independent).
Hence, assuming the error has a Gaussian tail (which is also far from obvious),
and recalling that the probability that a Gaussian variable deviates by more than $t$
standard deviations is roughly $\exp(-t^2)$, we obtain that the
probability of seeing a total deviation of more than $1\pm 0.1$ is at most $\exp(-n^{1/3})$, as required.

We remark that we also have an alternative and direct proof of Theorem~\ref{thm:informalsampling} that is
similar in nature to the proof of Theorem~\ref{thm:informalsamplingequator} (as described below), except it uses
the Grassmannian manifold; this proof, unfortunately, currently leads to a worse bound of
$\exp(-n^{1/4})$ (instead of $\exp(-n^{1/3})$) and is therefore omitted. It is quite possible
that this direct proof can be improved to obtain the tight $\exp(-n^{1/3})$ bound.

The proof of Theorem~\ref{thm:informalsamplingequator} will be
given in Section~\ref{sec:samplingequator}. It uses
the hypercontractive inequality for the sphere, applied in a fashion
similar to the one done by Kahn, Kalai, and Linial~\cite{KahnKL88},
as well as some basic properties of the Radon transform.
In order to demonstrate these ideas in a setting that might be more familiar to some readers,
we spend the remainder of this section on proving an analogous statement in the setting of the Boolean hypercube
$\{0,1\}^n$, and for simplicity just consider the case $t=n^{-1/3}$ (the general case is similar).

\myparagraph{Sampling statement for the Boolean cube}
Let $n$ be an even integer. For a vector $y \in \{0,1\}^n$ define $y^\perp = \{ z \in \{0,1\}^n ; {\rm HamDist}(y,z)=n/2 \}$
as the ``equator orthogonal to $y$".
Let $A \subseteq \{0,1\}^n$
be of measure $\mu(A):=|A|/2^n$ at least $\exp(-n^{1/3})$. Assume we choose a uniform $y \in \{0,1\}^n$,
and consider the fraction of points in $y^\perp$ that are contained in $A$. Then our goal
is to show that this fraction is in $(1 \pm n^{-1/3}) \mu(A)$ except with probability
at most $\exp(-n^{1/3})$.

As stated, this statement is actually \emph{false} due to a parity issue; this can be seen, e.g.,
by taking $A$ to be all points of even Hamming weight, a set of measure $1/2$. Then the fraction
of points in $y^\perp$ that are contained in $A$ is either $0$ or $1$ depending on the parity of
$y$.
Although the statement can be easily mended, in
the sequel we ignore this issue and proceed with an incomplete proof of the original incorrect statement.
We allow ourselves to do this because this parity issue does not arise in
the setting of the sphere, and the argument below becomes a valid proof there (with the necessary modifications, of course).

The above sampling statement can be stated in the following essentially equivalent way.
For any $A,B \subseteq \{0,1\}^n$ of measure at least $\exp(-n^{1/3})$,
\begin{align}
 \PP_{y \sim B, x \sim y^\perp}[x \in A] \in (1 \pm n^{-1/3}) \mu(A), \label{eq:hypercubestatement}
\end{align}
where the notation $x \sim E$ means that $x$ is distributed uniformly in the set $E$,
and the right hand side indicates the interval $[(1 -n^{-1/3})\mu(A), (1 + n^{-1/3})\mu(A)]$.
For a function $f:\{0,1\}^n \to \RR$, define its \emph{Radon transform} $R(f):\{0,1\}^n \to \RR$ as
$$ R(f)(y) := \EE_{x \sim y^\perp}[f(x)].$$
Define $f = 1_A / \mu(A)$ and $g = 1_B/\mu(B)$ to be the indicator functions of $A$ and $B$
normalized so that their expectations over a uniform input are $\EE_x[f(x)] = \EE_x[g(x)] = 1$. With this notation, Eq.~\eqref{eq:hypercubestatement}
becomes
\begin{align}
 \langle f , R(g) \rangle = \EE_x[f(x)R(g)(x)] \in 1 \pm n^{-1/3}.\label{eq:hypercubestatement2}
\end{align}
For a function $f:\{0,1\}^n \to \RR$, define its Fourier transform $\hat{f}:\{0,1\}^n \to \RR$
by $\hat{f}(w) := \EE_x[(-1)^{w\cdot x} f(x)]$. Then by the orthogonality of the Fourier
transform, Eq.~\eqref{eq:hypercubestatement2} can be written equivalently as
$$ \sum_w \hat{f}(w) \widehat{R(g)}(w) \in 1 \pm n^{-1/3}.$$

An easy direct calculation reveals that $R$ is diagonal in the Fourier basis.
(Alternatively, one can use Schur's lemma and the fact that $R$ commutes with translations.)
This calculation also reveals that the eigenvalue corresponding to $w \in \{0,1\}^n$ is $0$ whenever the Hamming
weight of $w$ is odd, $1$ when the Hamming weight of $w$ is $0$,
$$\frac{\binom{n-2}{n/2} - 2 \binom{n-2}{n/2-1} + \binom{n-2}{n/2-2}}{\binom{n}{n/2}}\approx -\frac{1}{n}$$
when $w$ is of Hamming weight $2$, approximately $\frac{1}{n^2}$ when $w$ is of Hamming weight $4$, etc.
We can therefore write
$$ \sum_w \hat{f}(w) \widehat{R(g)}(w) \approx
   \hat{f}(0)\hat{g}(0) - \frac{1}{n} \sum_{|w|=2} \hat{f}(w)\hat{g}(w) + \frac{1}{n^2} \sum_{|w|=4} \hat{f}(w)\hat{g}(w) - \cdots.$$
The first term is $\hat{f}(0)\hat{g}(0) = \EE_x[f(x)] \EE_x[g(x)] = 1$. Hence our goal is to bound the remaining terms by
$n^{-1/3}$. For simplicity, let us focus on the first term, and show that $\sum_{|w|=2} \hat{f}(w)\hat{g}(w)$
is at most $n^{2/3}$ in absolute value; one can similarly analyze the remaining terms and show that
their total contribution is similar.\footnote{This is where we
are cheating: the term $|w| = n$ can contribute a lot to this sum.} By using the Cauchy-Schwarz
inequality, we can bound this sum by $(\sum_{|w|=2} \hat{f}(w)^2)^{1/2} (\sum_{|w|=2} \hat{g}(w)^2)^{1/2}$.
The following lemma now completes the proof.

\begin{lemma}
Let $A \subseteq \{0,1\}^n$ be of measure $\mu$, and let $f = 1_A/\mu(A)$ be its (normalized) indicator function.
Then, for some universal constant $C>0$,
$$ \sum_{|w|=2} \hat{f}(w)^2 \le C (\log(1/\mu))^2.$$
\end{lemma}

Equivalently, the lemma says that if $X=(x_1,\ldots,x_n)$ is uniformly chosen from $A$,
then the sum over all pairs $\{i,j\}$ of the bias squared of $x_i \oplus x_j$ is at most $C(\log(1/\mu))^2$.
This can be seen to be essentially tight by taking, e.g., $A=\{x \in \{0,1\}^n; x_1=\cdots=x_{\log_2{1/\mu}}=0\}$.
This lemma is proven by applying the Bonami-Gross-Beckner hypercontractive inequality \cite{Bonami70, Gross75, Beckner75} in a way similar to that in~\cite{KahnKL88}.
Essentially the exact same lemma appears in~\cite{GavinskyKKRW07}, and is also described in
detail in the survey~\cite{deWolfSurvey}. We include a sketch of the proof, as later on we will have a similar
proof in the spherical setting (in Lemma~\ref{lem:KKL}).

\begin{proof}
The hypercontractive inequality for the Boolean cube states that for any $f: \{0,1\}^n \to \RR$,
and $1 \le p \le 2$,
$$ \big\| T_{\sqrt{p-1}} f \big\|_2 \le \|f \|_p $$
where $T_\rho$ is the noise operator with parameter $\rho$ (which is the operator
that is diagonal in the Fourier basis, and has eigenvalue $\rho^k$ for each Fourier basis function
of level $k$), and the $p$th norm is defined as $\|f\|_p = \EE_x[|f(x)|^p]^{1/p}$.
By plugging in our $f$ we obtain
\begin{align*}
\sum_{|w|=2} \hat{f}(w)^2 &\le \frac{1}{(p-1)^2} \sum_{w} (p-1)^{|w|} \hat{f}(w)^2 \\
& =  \frac{1}{(p-1)^2} \big\|T_{\sqrt{p-1}} f \big\|_2^2 \\
& \le \frac{1}{(p-1)^2} \|f\|_p^2 = \frac{1}{(p-1)^2} \mu^{-2(1-1/p)}.
\end{align*}
The lemma follows by optimizing over $1 \le p \le 2$.
\end{proof}

\section{Preliminaries}\label{sec:prelim}

\myparagraph{General}
Throughout the paper, by ``measurable" we mean Borel measurable. All logarithms are natural logarithms
unless otherwise specified.
We adopt the following convention for
denoting constants.
The letters $c, \tilde{c},
C, \tilde{C}$, etc.\ stand for various positive universal
constants, whose value may change from one line to the next. We
usually use upper-case $C$ to denote universal constants that we
think of as ``sufficiently large'', and lower-case $c$ to denote
universal constants that are ``sufficiently small''.

\myparagraph{Some manifolds and uniform distributions on them}
Write $S^{n-1} = \{ x \in \RR^n ; |x| = 1 \}$ for the unit sphere
in $\RR^n$. We denote by $\sigma$ the
uniform probability measure on $S^{n-1}$, i.e., the unique
rotationally-invariant probability measure on $S^{n-1}$ (see, e.g.,~\cite[Chapter I]{MilmanS86} for more
information on Haar measures).
We denote by $\cG_{n,m}$ the Grassmannian manifold, i.e., the manifold of all $m$-dimensional
subspaces in $\RR^n$, and we let $\sigma_{\cG}$ be the uniform distribution over it
(or, more formally, the unique rotationally-invariant probability measure on $\cG_{n,m}$).
We also consider the incidence manifold
$$ \cI_{n,m} = \left \{ (x, H) \in S^{n-1} \times \cG_{n,n-m} \, ; \, x  \in H \right \} \subset S^{n-1} \times \cG_{n,n-m}, $$
and let $\sigma_\cI$ be the uniform probability measure on it (or more precisely,
the unique rotationally-invariant probability measure on $\cI_{n,m}$).
We will implicitly use some basic properties of these manifolds and the uniform distributions on them;
for a rigorous discussion of the topic, see, e.g., Helgason~\cite[Chapter II]{Helgason99}.

\section{Communication Complexity}\label{sec:comm}

In this section we give a formal definition of the $\VSP$ problem, and derive the
main lower bound from the sampling statement. Our discussion in this section closely follows Raz's~\cite{Raz99}, hence we will occasionally
allow ourselves to be brief. We also assume some basic familiarity with randomized communication
complexity~\cite{kushilevitznisan-book}.

We start with the formal definition of $\VSP$. This is identical to the ${\cal P}_0$ problem defined
in~\cite{Raz99}.

\begin{definition}
Let $0 \le \vartheta < 1/\sqrt{2}$ be a parameter. In the $\VSP_\vartheta$ problem, Alice is given
an $n$-dimensional unit vector $u \in S^{n-1}$ and Bob is given a subspace $H \subset \RR^n$ of
dimension $n/2$. They are promised that either the distance of $u$ from $H$ is at most $\vartheta$
or the distance of $u$ from $H^\perp$ is at most $\vartheta$. Their goal is to decide which is the case.
\end{definition}

This problem was first defined by Kremer~\cite{Kremer95} and was shown to be a complete
problem for one-round quantum communication complexity. In particular, for any $0 \le \vartheta < 1/\sqrt{2}$,
$\VSP_\vartheta$ has an (almost immediate) quantum protocol communicating only $O(\log n)$
qubits in a single message from Alice to Bob. (Moreover, there is a matching $\Omega(\log n)$
lower bound.)

In terms of its classical (randomized, bounded-error) communication complexity, Raz~\cite{Raz99}
has shown that the problem has an $O(\sqrt{n})$ communication protocol, which we now
briefly describe. Assume Alice and Bob use their shared randomness
to pick a sequence of unit vectors chosen uniformly from $S^{n-1}$, $v_1,v_2,\ldots$. Alice
looks for the vector $v_i$ with the maximal inner product $v_i \cdot u$ among the first $2^{C \sqrt{n}}$
unit vectors, and sends the index
$i$ to Bob, who decides on the output based on which of $H$ and $H^\perp$ is closer to $v_i$.
The protocol clearly requires only $O(\sqrt{n})$ bits
of communication, and moreover, the output produced by Bob is correct with high probability
(essentially since the projection squared of $v_i$ on $H$ (or $H^\perp$) gets an addition of $n^{-1/2}$
due to the high inner product with $u$, which is sufficient to noticeably affect Bob's answer since the standard
deviation of the projection squared is of order $n^{-1/2}$). Using Newman's theorem,
the shared randomness can be replaced with private randomness by only communicating
an extra $O(\log n)$ bits (which is negligible).
For a more detailed proof, see Theorem 3.8 in~\cite{Raz99}.

However, no lower bound better than logarithmic was previously known.
Our main result is an $\Omega(n^{1/3})$ lower bound on the randomized communication complexity of the problem $\VSP_0$ (which is the problem described in the
introduction). One minor caveat here is that this lower bound holds only for protocols that are ``measurable,"
in the sense that the functions describing the behavior of the players need to be measurable.
Clearly, increasing $\vartheta$ can only make the problem harder, hence
our lower bound also apply to any $0 < \vartheta < 1/\sqrt{2}$. Moreover,
as we shall see below, there is no need to assume measurability in the case
$\vartheta > 0$.

Another point to note is that the number of possible inputs to
$\VSP$ is infinite. Although there is nothing terribly wrong with this, in the
standard communication complexity model problems are supposed to have inputs that are taken from a finite set.
This can be easily achieved by specifying the inputs using
an $n$-dimensional vector (for Alice) together with an $n \times n/2$ matrix (for Bob) each of whose entries
is described by $O(\log n)$ bits. We denote this problem by $\widetilde{\VSP}$.
Since this is a restriction of $\VSP$, we clearly still have a one-way $O(\log n)$ qubit protocol.
Next, notice that for any $0 < \vartheta < 1/\sqrt{2}$, we can convert any protocol for $\widetilde{\VSP}_\vartheta$
into a protocol for $\VSP_0$ by simply rounding the coordinates
of the inputs. Moreover, the resulting $\VSP_0$ protocol is clearly measurable since its
input space is partitioned into a finite number of simple sets, and the protocol's
behavior is completely determined on each of these simple sets.
We therefore obtain a lower bound of $\Omega(n^{1/3})$ on the randomized communication complexity of $\widetilde{\VSP}_\vartheta$
for any $0 < \vartheta < 1/\sqrt{2}$. Notice that the problem's input size is $m=O(n^2 \log n)$,
and hence in terms of the input size, our lower bound is $\Omega((m/\log m)^{1/6})$.
Finally, since $\widetilde{\VSP}_\vartheta$ is a restriction of $\VSP_\vartheta$,
we also obtain a lower bound of $\Omega(n^{1/3})$ on the randomized communication complexity of $\VSP_\vartheta$
for any $0 < \vartheta < 1/\sqrt{2}$, without the measurability assumption.
We summarize this discussion in the following theorem, which we then proceed
to prove.

\begin{theorem}
Any measurable randomized (bounded-error) protocol for $\VSP_0$ requires $\Omega(n^{1/3})$
bits of communication. As a result, we obtain that for all $0 < \vartheta < 1/\sqrt{2}$,
the randomized communication complexity of both $\VSP_\vartheta$ and $\widetilde{\VSP}_\vartheta$
is $\Omega(n^{1/3})$ (without any measurability assumption).
\end{theorem}
\begin{proof}
As described above, it suffices to prove the lower bound on $\VSP_0$.
Fix an arbitrary randomized protocol communicating at most $D$ bits,
and assume that it solves $\VSP_0$ with error probability at most $1/3$ on all legal inputs.
(The argument applies to any error probability smaller than $1/2$ by a standard amplification technique.)
Our goal is to lower bound $D$.

Recall the definition of $\cI_{n,n/2}$ and the uniform distribution $\sigma_{\cI}$ on it,
given by a uniformly chosen subspace $H$ of dimension $n/2$ and a uniformly chosen unit vector $u$ in $H$.
We also define the set $\bar{\cI}_{n,n/2}$ as the set of all pairs
$(x, H) \in S^{n-1} \times \cG_{n,n/2}$ such that $x \in H^\perp$,
and let $\bar{\sigma}_{\cI}$ be the uniform distribution on it,
given by a uniformly chosen subspace $H$ of dimension $n/2$ and a uniformly chosen unit vector $u$ in $H^\perp$.

We consider the following two quantities. The first is the probability that
the protocol incorrectly outputs ``$u$ not in $H$" when the inputs are chosen from $\sigma_{\cI}$.
The second is the probability that the protocol incorrectly outputs ``$u$ in $H$" when the inputs
are chosen from $\bar{\sigma}_{\cI}$.
By our assumption, each of these quantities is at most $1/3$, and hence their sum is at most $2/3$.
By linearity there exists a way to fix the random string used by the protocol such that the
resulting deterministic protocol also satisfies that the sum of these two quantities is at most $2/3$.
{}From now on we consider that deterministic protocol.

As is well known, such a deterministic protocol induces a partition of
$S^{n-1} \times \cG_{n,n/2}$ into $2^D$ rectangles, i.e., measurable sets of the form $A \times B$ where
$A \subseteq S^{n-1}$ and $B \subseteq \cG_{n,n/2}$, where each rectangle is labelled
with ``in" or ``not in", corresponding to the protocol's output on inputs from
this rectangle. In order to analyze this partition, we use the following
lemma, which follows easily from our main sampling theorem, as will be shown in
Section~\ref{sec:maintechnical}.

\begin{lemma}\label{lem:main2}
Suppose that $A \subseteq S^{n-1}$ and $B \subseteq \cG_{n,n/2}$
are measurable sets with
$$ \sigma(A) \geq C \exp(-c n^{1/3}), \qquad \sigma_{\cG}(B) \geq C \exp(-c n^{1/3}) $$
for some universal constants $c,C>0$. Then,
$$ \sigma_{\cI} \left( (A \times B) \cap \cI_{n,n/2} \right) \geq 0.8\, \sigma(A) \sigma_{\cG}(B).$$
\end{lemma}

As a result, we obtain that for \emph{all} measurable sets $A \subseteq S^{n-1}$ and $B \subseteq \cG_{n,n/2}$,
\begin{align}\label{eq:commfinal1}
 \sigma_{\cI} \left( (A \times B) \cap \cI_{n,n/2} \right) \geq 0.8 \, \sigma(A) \sigma_{\cG}(B) - C \exp(-c n^{1/3}).
\end{align}
By simply replacing $H$ with $H^\perp$ we also obtain that
\begin{align}\label{eq:commfinal2}
 \bar{\sigma}_{\cI} \left( (A \times B) \cap \bar{\cI}_{n,n/2} \right) \geq 0.8 \, \sigma(A) \sigma_{\cG}(B) - C \exp(-c n^{1/3}).
\end{align}
We now sum the inequalities~\eqref{eq:commfinal1} over all rectangles $A\times B$ that are labelled with
``not in" and the inequalities~\eqref{eq:commfinal2} over all rectangles labelled with
``in". Our assumption above says precisely that the left hand side is at most $2/3$.
The right hand side is exactly $0.8 - 2^D \cdot C \exp(-c n^{1/3})$. Rearranging, we obtain
that $2^D \ge c \exp(c n^{1/3})$, as required.
\end{proof}

\section{Sampling Sets by Equators}\label{sec:samplingequator}

In this section we prove one of the main components of our proof, namely, a sampling theorem
using equators: we show that any (not too small) subset $A$ of the sphere $S^{n-1}$
is sampled well by a randomly chosen equator (where an equator is the intersection of $S^{n-1}$
with an $(n-1)$-dimensional subspace). See Figure~\ref{fig:equator}.

\begin{theorem}\label{thm:mainequator}
Let $A \subseteq S^{n-1}$ be a measurable set. Assume $H$ is a uniformly chosen $(n-1)$-dimensional subspace. Then, for any
$0<t<1$, the probability that
$$ \left| \frac{\sigma_H(A \cap H)}{\sigma(A)} - 1 \right| \ge t $$
is at most $C \exp(-c n t / \log (2/\sigma(A)))$ for some universal constants $C,c>0$,
where $\sigma_H$ denotes the uniform probability measure on the sphere $H \cap S^{n-1}$.
\end{theorem}

In the rest of this section, we actually prove the following more symmetric statement, from which
Theorem~\ref{thm:mainequator} follows as described below. Here we denote by $\cV_n$ the manifold of all
pairs of orthogonal vectors,
$$ \cV_{n} = \left \{ (x, y) \in S^{n-1} \times S^{n-1} \, ; \, x \cdot y = 0 \right \}
$$
and we let $\sigma_\cV$ denote the uniform probability measure over $\cV_n$.

\begin{theorem}\label{thm:mainsymmetric}
Suppose $f, g: S^{n-1} \rightarrow [0, \infty)$ are bounded measurable functions
with $\int_{S^{n-1}} f \d \sigma = \int_{S^{n-1}} g \d \sigma = 1$ and set
 $$ s = \log (2 \| f \|_{\infty})  \cdot \log (2 \| g \|_{\infty} ).$$
 Then, when $s \leq c n$,
$$
 \left| \int_{\cV_n} f(x) g(y) \d \sigma_\cV(x,y)  - 1 \right| \leq \frac{C s}{n},
$$
where $C,c  > 0$ are universal constants.
\end{theorem}

In order to derive Theorem~\ref{thm:mainequator}, let $E$ be the set of all $y \in S^{n-1}$
for which the subspace $y^\perp \subset \RR^n$ orthogonal to $y$ satisfies
$$ \frac{\sigma_{y^\perp}(A \cap y^\perp)}{\sigma(A)} \ge 1+t. $$
Let $f=1_A/\sigma(A)$ and $g=1_E / \sigma(E)$ be the normalized indicator functions of $A$ and $B$, respectively.
Then it follows that
$$ \int_{\cV_n} f(x) g(y) \d \sigma_\cV(x,y) \ge 1+t$$
since the left hand side is exactly the average of $\sigma_{y^\perp}(A \cap y^\perp) / \sigma(A)$ over $y$ chosen uniformly from $E$.
Hence by Theorem~\ref{thm:mainsymmetric},
$$ t \le \frac{C \log (2/\sigma(A)) \log(2/\sigma(E))}{n}.$$
Rearranging, we obtain that
$$\sigma(E) < C \exp(-cnt/\log (2/\sigma(A))).$$
Repeating a similar argument for the lower bound, we obtain Theorem~\ref{thm:mainequator}.

Our proof of Theorem~\ref{thm:mainsymmetric} resembles a small jigsaw puzzle, in which all of the
pieces are known mathematical constructions that have to be put in
place in order to yield a proof. Therefore most of this section is
devoted to a brief summary of standard mathematical material, such
as some basic features of spherical harmonics, the Laplacian,
log-Sobolev inequalities, hypercontractivity,
growth of $L^p$ norms of eigenfunctions, and the Radon transform and its eigenvalues.

\myparagraph{Spherical harmonics} We write $L^2(S^{n-1})$ for the space of all square-integrable
functions on $S^{n-1}$. For $U \in SO(n)$ and $f \in L^2(S^{n-1})$  denote
$$
 U(f)(x) = f(U^{-1} x) \quad \quad \quad \quad (x \in S^{n-1}).
$$
We say that $U(f)$ is the rotation of $f$ by $U$. For any integer $k
\geq 0$, there is a special finite-dimensional subspace $\cS_k
\subset L^2(S^{n-1})$ of smooth functions called the space of
``spherical harmonics of degree $k$." For instance, $\cS_0$
is the one-dimensional space of constant functions.
More generally, $\cS_k$ is defined as the
restriction to the sphere of all harmonic, homogenous polynomials of
degree $k$ in $\RR^n$. See, e.g., M\"uller~\cite{Muller66} or
Stein and Weiss~\cite{SteinW71} for a quick introduction and for more
information on spherical harmonics. The space $\cS_k$ is invariant
under rotations and hence provides a representation of $SO(n)$. This representation is
known to be irreducible, that is, for any subspace $E \subseteq \cS_k$ that is
invariant under rotations, we necessarily have
$$ E = \{ 0 \} \quad \quad \quad \text{or} \quad \quad \quad E =
\cS_k. $$
Moreover, these representations in $\cS_k$ for $k=0,1,\ldots$ are known to be inequivalent;
this follows, e.g., from the fact that their dimensions (given by $\binom{n+k-1}{n-1}-\binom{n+k-3}{n-1}$)
are all different (assuming $n \ge 3$).
 Elements of $\cS_k$ are orthogonal to elements of
$\cS_\ell$ for $k \neq \ell$. We denote by $Proj_{\cS_k}$ the
orthogonal projection operator onto $\cS_k$ in $L^2(S^{n-1})$. Then
any function $f \in L^2(S^{n-1})$ may be decomposed  as
$$
 f = \sum_{k=0}^{\infty} Proj_{\cS_k} f
$$
where the sum converges in $L^2(S^{n-1})$. This decomposition
of a function on $S^{n-1}$ is analogous to the
decomposition of a function on the Boolean hypercube into
Fourier levels.

\myparagraph{Laplacian}
Write $C^{\infty}(S^{n-1})$ for the space of infinitely differentiable functions
on $S^{n-1}$. For a function $f \in C^{\infty}(S^{n-1})$ and $x \in S^{n-1}$, we define
\begin{equation}  (\triangle f)(x) = \sum_{i=1}^{n-1} \left.
\frac{d}{dt^2} f \left( (\cos t) x + (\sin t) e_i \right) \right|_{t=0},
\label{laplace}
\end{equation} where $e_1,\ldots,e_{n-1}$ is an orthonormal basis of
$x^{\perp}$. Notice that for any orthogonal $x, y \in S^{n-1}$, the curve $t \mapsto (\cos t) x + (\sin t) y$ draws a great
circle on $S^{n-1}$, that visits $x$ at $t = 0$, and its tangent
vector at $t = 0$ is the vector $y$.
The right hand side of \eqref{laplace} does not depend
on the choice of the orthonormal basis $e_1,\ldots,e_{n-1}$. The
operator $\triangle$, acting from $C^{\infty}(S^{n-1})$ to itself,
is called the \emph{spherical Laplacian}.

\medskip One computes (see, e.g.,
\cite{SteinW71}) that for any $k \geq 0$ and $\vphi_k \in \cS_k$,
\begin{equation}\label{eq:laplacianevalues}
\triangle \vphi_k = -\lambda_k \vphi_k
\end{equation}
where
$$
\lambda_k = k (k + n - 2).
$$
The Laplacian thus has a complete system of orthonormal
eigenfunctions in $L^2(S^{n-1})$ (even though the Laplacian is
defined only for smooth functions and not in the entire space
$L^2(S^{n-1})$).

\myparagraph{Noise operator}
The noise operators on $S^{n-1}$ are
$$ U_{\rho} = \rho^{-\triangle} \quad \quad \quad \quad (0 \leq \rho \leq 1). $$
A priori, these operators are defined, say, on the dense space of
finite linear combinations of spherical harmonics.  Since the norm
of $U_{\rho}$ does not exceed one, we may uniquely extend $U_{\rho}$
to a self-adjoint operator
$ U_{\rho} : L^2(S^{n-1}) \rightarrow L^2(S^{n-1}) $
of norm one. From~\eqref{eq:laplacianevalues} it follows that for
any $k \geq 0$ and $\vphi_k \in \cS_k$,
$$ U_{\rho} \vphi_k = \rho^{\lambda_k} \vphi_k.$$

\myparagraph{Hypercontractivity}
We proceed with a short review of hypercontractivity, a subject going
back to Nelson~\cite{Nelson66}. For $p \geq 1$ and for a measurable function
$f: S^{n-1} \rightarrow \RR$ we write $ \| f \|_p =
\left(\int_{S^{n-1}} |f|^p \d \sigma \right)^{1/p} $ for the
$L^p$-norm of $f$. The hypercontractive inequality states that for any $1 \le p \le q$,
and any function $f \in L^p(S^{n-1})$,
\begin{equation}
\| U_{\rho} f \|_{q} \leq \| f \|_p \quad \quad \quad \quad
\text{for} \ \ 0 \le \rho \leq \left( \frac{p-1}{q-1}  \right)^{1 /
(2n-2)}. \label{eq:hyper}
\end{equation}

We now briefly describe how one proves such an inequality.
By differentiating with respect to $p$ and $q$,
Gross~\cite{Gross75} showed that hypercontractive inequalities such as the one
above are directly equivalent to so-called \emph{log-Sobolev inequalities}.
Indeed, a common  technique for proving
hypercontractive inequalities  is by proving the analogous
log-Sobolev inequality (as the latter is often cleaner and easier to
work with). More specifically, for our hypercontractive inequality~\eqref{eq:hyper},
the equivalent log-Sobolev inequality turns out to be
\begin{equation}
 \int_{S^{n-1}} f^2(x) \log \frac{f^2(x)}{\int f^2(y) \d \sigma(y) } d \sigma(x)
 \leq \frac{1}{n-1} \int_{S^{n-1}} |\nabla f(x)|^2 \d \sigma(x) \label{eq_1058}
\end{equation}
for any smooth $f: S^{n-1} \rightarrow \RR$ where $\nabla f$ denotes the gradient of $f$.
Finally, this (tight) inequality was proven by Rothaus~\cite{Rothaus86}.

We note that a slightly
weaker inequality, in which $\frac{1}{n-1}$ is replaced by
$\frac{1}{n-2}$ (leading to a corresponding worsening of the exponent in~\eqref{eq:hyper} from $1/(2n-2)$ to $1/(2n-4)$),
follows from the elegant \emph{Bakry-{\'E}mery criterion}
(see~\cite{BakryE85}, or e.g.,~\cite{BakryL06}). This criterion states
that a log-Sobolev inequality holds for any connected manifold whose Ricci curvature
is uniformly bounded from below by some positive constant. In our very special
case, the manifold is $S^{n-1}$, whose Ricci curvature is constantly $n-2$,
leading to~\eqref{eq_1058} with the slightly weaker
constant $\frac{1}{n-2}$. This slightly weaker version certainly
suffices for all of our needs in this paper.

\medskip Kahn, Kalai, and Linial~\cite{KahnKL88} realized that
hypercontractive inequalities such as \eqref{eq:hyper} imply certain
bounds on the growth of $L^p$ norms of the Laplacian eigenfunctions.
Although they focused on the Boolean hypercube, their idea can be applied
in much greater generality, and in particular to the sphere.
Indeed, suppose $\vphi_k \in \cS_k$ for some $k \geq 0$.  Then
$U_{\rho} \vphi_k = \rho^{\lambda_k} \vphi_k$. {}From \eqref{eq:hyper},
for any $1 \leq p \leq q$,
\begin{equation}   \| \vphi_k \|_q \leq
\left( \frac{q-1}{p-1} \right)^{\lambda_k / (2n-2)} \| \vphi_k \|_p.
\label{eq_1638}
\end{equation}
For large $n$ and fixed $k$, we have $\lambda_k / (2n - 2) \approx
k/2$. In this case, the bound \eqref{eq_1638} roughly says that
for any $t$, the set of points $x \in S^{n-1}$ where $|\vphi_k| \geq
t \| \vphi_k \|_1$ has measure at most  $C \exp(-c t^{2/k})$. The
following lemma runs in a similar vein, and provides an upper bound on the mass
that the indicator function of a set can have on each level of the spherical
harmonics decomposition.

\begin{lemma}\label{lem:KKL}
Suppose $f: S^{n-1} \rightarrow \RR$
satisfies $\| f \|_1 = 1$ and $\| f \|_{\infty} \leq M$. Then, for
any $k \geq 1$,
$$
 \left \| Proj_{\cS_k} f  \right \|_2 \leq \left( e \cdot \max\left( 1, \frac{\log M}{\lambda_k / (2n-2)}\right)
 \right)^{\lambda_k / (2n-2)}.
$$
\end{lemma}

\begin{proof}
First, note that for any $p \ge 1$,
$$
 \| f \|_{p} = \left( \int_{S^{n-1}} |f|^p \d \sigma \right)^{1/p}
\leq \left( M^{p-1} \int_{S^{n-1}} |f| \d \sigma \right)^{1/p} =
M^{(p-1) / p} \le M^{p-1}.
$$
In particular, since $\| Proj_{\cS_k} f \|_2 \leq \|f\|_2 \le M$, we obtain
that the lemma holds whenever $\lambda_k > (2n-2) \log M$. So assume from now on that
$\lambda_k \le (2n-2) \log M$.
We use \eqref{eq:hyper} for $q=2$ and obtain that for any $1 \le p \le 2$,
$$  \|  U_{\rho} f \|_{2} \leq \| f \|_{p} \leq M^{p-1} \quad \quad \text{for} \ \rho
= \left( p-1 \right)^{1/(2n-2)}. $$ Projecting to $\cS_k$, we see
that for any $1 \leq p \leq 2$,
$$
  (p-1)^{\frac{\lambda_k}{2n-2}} \left \|
Proj_{\cS_k}f \right \|_{2} = \|  Proj_{\cS_k} (U_{\rho} f) \|_{2}
\leq \| U_{\rho} f \|_2 \leq  M^{p-1}.
$$
We complete the proof by choosing $ p = 1 + \frac{\lambda_k}{(2n-2)
\log M} \leq 2$.
\end{proof}

\myparagraph{Radon transform}
Recall that for $\theta \in S^{n-1}$ we write $\sigma_{\theta^{\perp}}$
 for the uniform probability measure on
the sphere $S^{n-1} \cap \theta^{\perp}$. Then the spherical Radon
transform $R(f)$ of an integrable function  $f: S^{n-1} \rightarrow \RR$ is
defined as
$$ R (f)(\theta) = \int_{S^{n-1} \cap \theta^{\perp}} f(x) \d \sigma_{\theta^{\perp}}(x), \quad \quad \quad \quad (\theta \in S^{n-1}). $$
So $R(f)$ is simply the average of $f$ on the equator of vectors orthogonal to $\theta$.
Observe that for
functions $f, g \in L^2(S^{n-1})$, we have
\begin{align}\label{eq:ortohgonalhelga}
 \int_{\cV_{n}} f(x) g(y) \d \sigma_{\cV}(x,y) =
 \int_{S^{n-1}} f(x) Rg(x) \d \sigma(x).
\end{align}
This equality describes the intuitive fact that integrating uniformly over all
orthogonal pairs $(x,y)$ is the same as integrating uniformly over $x$,
and then uniformly over all $y$ in the orthogonal complement of $x$. See, e.g.,
Helgason~\cite[Chapter II]{Helgason99} for a more formal derivation.

Define a sequence of numbers $( \mu_k )_{k=0,1,\ldots}$ as follows.
Suppose $X = (X_1,\ldots,X_{n-1})$ is a random vector that is
uniformly distributed in $S^{n-2}$. For an even $k \geq 0$ denote
$$
 \mu_k = (-1)^{k/2} \EE [X_1^k],
$$
and for odd $k$ set $\mu_k = 0$. We now show that $\cS_k$ are the eigenspaces
of $R$ with $\mu_k$ being the corresponding eigenvalues.

\begin{lemma}\label{lem_1809}
For any $k \geq 0$ and
$\vphi_k \in \cS_k$,
$$
 R(\vphi_k) = \mu_k \vphi_k. $$
\end{lemma}

\begin{proof}
The Radon transform clearly commutes with rotations. Therefore, because the $\cS_k$'s give rise
to inequivalent irreducible representations, Schur's lemma implies that $R$ must have the $\cS_k$'s as its
eigenspaces. We briefly recall the proof of this standard representation-theoretic fact.
Consider the restriction $R_{k,j}$ of $Proj_{\cS_j} R$ to an operator
from $\cS_k$ to $\cS_j$ for some $k,j \ge 0$. Our goal is to show that $R_{k,j}$ is zero
whenever $k\neq j$ and a multiple of the identity otherwise.
Since $R_{k,j}$ commutes with the action of $SO(n)$,
and $\cS_k$ is irreducible, we have that $\ker R_{k,j}$ is either all of $\cS_k$ or $\{0\}$.
In the former case $R_{k,j}=0$ and we are done, so assume the latter case.
By the same argument the image of $R_{k,j}$ is either all of $\cS_j$ or $\{0\}$,
and since we assumed $R_{k,j}\neq 0$, it must be the former. Hence $R_{k,j}$ is an isomorphism between the
representation on $\cS_k$ and on $\cS_j$, which is impossible when
$k \neq j$ since we know that $\cS_k$ and $\cS_j$ are inequivalent representations.
So assume $k=j$, and let $\lambda \in \RR$ be an arbitrary eigenvalue of
$R_{k,k}$ (there exists such an eigenvalue since $R_{k,k}$ is a symmetric operator). Then the kernel
of $\lambda I - R_{k,k}$ must also be either all of $\cS_k$ or $\{0\}$;
the latter is impossible since $\lambda$ is an eigenvalue, hence we necessarily have $R_{k,k} = \lambda I$.

Our next goal is to show that the $\mu_k$'s are the corresponding
eigenvalues. Fix some arbitrary $e \in S^{n-1}$. For $k \ge 0$, we define the function $f_k:S^{n-1} \to \RR$
by
$$ f_k(x) = G_k(x \cdot e) \quad \quad \quad \quad (x \in S^{n-1})
$$
where $G_k:[-1,1] \to \RR$ is the Gegenbauer polynomial (see, e.g.,
M\"uller~\cite{Muller66}),
$$
 G_k(t) = \EE \left(t + i X_1 \sqrt{1 - t^2}  \right)^k.
$$
Here, $i^2 = -1$ and $X = (X_1,\ldots,X_{n-1})$ is a random
vector that is distributed uniformly over the sphere $S^{n-2}$.
The function $f_k$ is known to be a spherical harmonic of degree $k$, i.e., in $\cS_k$~\cite{Muller66},
and by our above discussion, must be an eigenfunction of $R$, i.e., $Rf$ is proportional to $f$.
{}From the definition of the Radon transform,
$$ (R f)(e) = G_k(0) \quad \quad \quad \text{and} \quad \quad \quad
f(e) = G_k(1) = 1. $$ We conclude that $G_k(0)$ is the eigenvalue
corresponding to $\cS_k$. It remains to notice that
$G_k(0)$ vanishes for odd $k$ and equals  $(-1)^{k/2} \EE
X_1^k$ for even $k$, and hence equals $\mu_k$ for all $k$.
\end{proof}

The next technical lemma gives upper bounds on the eigenvalues $\mu_k$.
\begin{lemma}\label{lem:muk}
Suppose $n \geq 10$. Then, the sequence $|\mu_0|, |\mu_2|, |\mu_4|, \ldots$ is
non-increasing, and moreover, for all $k \geq 1$,
$$ |\mu_k| \leq \left( C \frac{k}{n} \right)^{k/2}.$$
\end{lemma}

\begin{proof}
The first claim follows immediately from the fact that $|X_1| \le 1$ and $|\mu_{2k}| = \EE[|X_1|^{2k}]$.
For the second claim, notice that the density of $X_1$ is proportional to $(1 - x^2)^{(n-4)/2}$
for $x \in [-1,1]$, and vanishes outside this interval. Hence, our goal is to
prove that for all even $k \ge 2$,
$$ \int_{-1}^1 x^k (1 - x^2)^{(n-4)/2} \d x \leq \left( C \frac{k}{n}
\right)^{k/2} \int_{-1}^1 (1 - x^2)^{(n-4) / 2} \d x. $$
The integral on the
right hand side is at least $c / \sqrt{n}$ (this is true even for
the integral from $-1/\sqrt{n}$ to $1/\sqrt{n}$). The integral on
the left hand side may be estimated as follows:
$$
\int_{-1}^1 x^k (1 - x^2)^{\frac{n-4}{2}} \d x \leq \int_{-1}^1 x^k e^{- \frac{n-4}{2} x^2}  \d x
\leq \int_{-\infty}^\infty x^k e^{- \frac{n-4}{2} x^2}  \d x.
$$
The latter integral is exactly the $k$th moment of a normal variable with mean $0$ and variance
$1/(n-4)$, times the missing normalization factor of $\sqrt{2 \pi / (n-4)}$. A standard fact
is that for even $k$ this moment is
$$(n-4)^{-k/2} \cdot (k-1)!! \le \left(\frac{k}{n-4}\right)^{k/2}$$
where $(k-1)!! = (k-1)(k-3)\cdots 1$. The lemma follows.
\end{proof}

\begin{proof}[Proof of Theorem~\ref{thm:mainsymmetric}]
It suffices to prove the theorem under the assumption that $n \geq 10$ (otherwise
there is no $s \leq cn$, for a sufficiently small universal constant $c > 0$). By Lemma~\ref{lem_1809}
and~\eqref{eq:ortohgonalhelga},
$$ \int_{\cV_n} f(x) g(y) \d \sigma_\cV(x,y) =  \int_{S^{n-1}}
f R \left(g \right) \d \sigma = \sum_{k=0}^{\infty} \mu_k
 \int_{S^{n-1}} Proj_{\cS_k} \left( f \right) Proj_{\cS_k}
\left( g \right) \d \sigma.
$$
Note that $\mu_0 = 1$ and $Proj_{\cS_0} \left( f \right) \equiv
Proj_{\cS_0} \left( g \right) \equiv 1$. Therefore, by the Cauchy-Schwarz inequality,
$$
 \left| \int_{\cV_n} f(x) g(y) \d \sigma_\cV(x,y)
 \, - \, 1 \right| \leq
 \sum_{k=1}^{\infty} |\mu_{2k}|  \| Proj_{\cS_{2k}} f \|_2 \| Proj_{\cS_{2k}} g \|_2.
$$

We will prove the theorem by showing that the latter sum is at most $C \alpha \beta / n$,
where $\alpha = \log (2 \|f\|_\infty)$ and $\beta = \log (2 \|g\|_\infty)$.
Observe that $\alpha,\beta \ge 1/2$ and recall our assumption that $\alpha \beta$ is at most $cn$.
We start by analyzing the part of the sum in
which $k$ runs from $1$ to $T-1$ where $T=\floor{\delta n}$ for some sufficiently small constant $\delta>0$.
Using Lemmas~\ref{lem:KKL} and~\ref{lem:muk}, we have the bounds
\begin{align*}
|\mu_{2k}| &\le \left(\frac{Ck}{n}\right)^k, \\
\| Proj_{\cS_{2k}} f \|_2 &\le  \left( C \max\left( 1, \frac{\alpha}{k}\right) \right)^{\lambda_{2k} / (2n-2)},
\end{align*}
and similarly for $g$ with $\beta$.
Therefore,
$$
\sum_{k=1}^{T-1} |\mu_{2k}|  \| Proj_{\cS_{2k}} f \|_2 \| Proj_{\cS_{2k}} g \|_2
\le
\sum_{k=1}^{T-1} \left(\frac{Ck}{n}\right)^k
     \left( C \max\left( 1, \frac{\alpha}{k}\right) \right)^{\lambda_{2k} / (2n-2)}
     \left( C \max\left( 1, \frac{\beta}{k}\right) \right)^{\lambda_{2k} / (2n-2)}.
$$

The term $k=1$ is at most
$$ \frac{C \alpha \beta}{n}.$$
We will now show that the terms in the latter sum decay geometrically,
and hence we can also bound the sum by $ C \alpha \beta / n$.
To this end, first notice that
$$ \left(\frac{C(k+1)}{n}\right)^{k+1} \Big/ \left(\frac{Ck}{n}\right)^k =
\frac{C(k+1)}{n}  \cdot \left(\frac{k+1}{k}\right)^k \le
\frac{\tilde{C}k}{n}.$$
Second,
\begin{align*}
\left( C \max\left( 1, \frac{\alpha}{k+1}\right) \right)^{\lambda_{2k+2} / (2n-2)} \Big/
\left( C \max\left( 1, \frac{\alpha}{k}\right) \right)^{\lambda_{2k} / (2n-2)}
&\le
\left( C \max\left( 1, \frac{\alpha}{k}\right) \right)^{(\lambda_{2k+2}-\lambda_{2k}) / (2n-2)} \\
&=
\left( C \max\left( 1, \frac{\alpha}{k}\right) \right)^{1+\frac{4k+1}{n-1}} \\
&\le \tilde{C} \max\left( 1, \frac{\alpha}{k}\right).
\end{align*}
Hence the ratio between the term for $k+1$ and that for $k$ is at most
$$ C \frac{k}{n} \max\left( 1, \frac{\alpha}{k}\right) \max\left( 1, \frac{\beta}{k}\right) \le \frac{1}{2},$$
as $k \leq \delta n$, once we choose $\delta$ to be a sufficiently small positive universal constant. This implies
that we can upper bound the sum from $1$ to $T-1$ by $C \alpha \beta/n$, as required.

It remains to analyze the less significant part of the sum, in which $k$ runs from $T=\floor{\delta n}$ to infinity.
Then, by Lemma~\ref{lem:muk} and another application of Cauchy-Schwarz,
\begin{align*}
\sum_{k=T}^{\infty} |\mu_{2k}|  \| Proj_{\cS_{2k}} f \|_2 \| Proj_{\cS_{2k}} g \|_2 &\le
|\mu_{2T}| \sum_{k=T}^{\infty} \| Proj_{\cS_{2k}} f \|_2 \| Proj_{\cS_{2k}} g \|_2 \\
&\le |\mu_{2T}| \|f\|_2 \|g\|_2 \\
&\le \exp(\alpha+\beta - cn) \\
&\le \frac{C}{n} \le \frac{\tilde{C} \alpha \beta}{n},
\end{align*}
under the legitimate assumption that $\alpha \beta \leq \tilde{c} n$. We conclude that the entire
sum is bounded by $C \alpha \beta / n$. \end{proof}

\section{Sampling Sets by Lower Dimensional Subspaces}\label{sec:maintechnical}

Our next step is to iterate Theorem~\ref{thm:mainsymmetric},
using a certain martingale process, in order to
obtain a corresponding theorem for the Grassmannian.
The constants $0.1$ and $9/10$ appearing below do not play any special role
and can be replaced with any other constants (as long as the former is positive and the
latter is smaller than $1$).

\begin{theorem}\label{thm:mainsamplinglowdim}
Let $1 \leq m \leq 9n / 10$.
Suppose that $A \subseteq S^{n-1}$ is a measurable set with
$ \sigma(A) \geq C \exp(-c n^{1/3})$.
Assume that $H$ is a uniformly chosen $(n-m)$-dimensional subspace. Then,
$$ \left| \frac{\sigma_H(A \cap H)}{\sigma(A)} - 1 \right| < 0.1 $$
except with probability at most $C \exp(-c n^{1/3})$.
Here, $c, C> 0$ are universal constants.
\end{theorem}

We start with a few technical lemmas.
The first one below bounds the moments of a random variable that has an exponentially decaying tail around $1$.
We will apply it with random variables whose expectation is very close to $1$.
\begin{lemma}\label{lem:moments}
Let $R, \delta > 0$ and let  $Z$ be a non-negative random variable satisfying
that for any $t \ge 0$,
$$
\PP(|Z - 1| \geq t) \leq R \exp(-t/\delta).
$$
Then, for any $2 \leq \ell \leq (2\delta)^{-1}$,
$$
\EE [Z^{\ell}] \leq 1 + \ell \, \EE[Z-1] + 2 R (\ell \delta)^2.
$$
\end{lemma}
\begin{proof}
Using the Taylor expansion, we have that for any $x \ge -1$,
\begin{align*}
 (1+x)^\ell &= 1 + \ell x + \sum_{k=2}^{\floor{\ell}-1} \binom{\ell}{k} x^k + \binom{\ell}{\floor{\ell}} (1+\xi)^{\ell - \floor{\ell}} x^{\floor{\ell}}  \\
   &\le 1 + \ell x + \sum_{k=2}^{\floor{\ell}} \frac{\ell^k}{k!} |x|^k + \frac{\ell^{\floor{\ell}}}{\floor{\ell}!} |x|^{\floor{\ell}+1},
\end{align*}
where $\xi$ is some real number between $x$ and $0$.
Next, for any $k \ge 1$,
\begin{align}
\EE[|Z-1|^k] &= \int_0^\infty \PP[|Z-1|^k \ge t] \d t \nonumber \\
&= \int_0^\infty k t^{k-1} \PP[|Z-1| \ge t] \d t \nonumber \\
&\le R \, k \, \int_0^\infty t^{k-1} \exp(-t/\delta) \d t = R \cdot k! \cdot \delta^k. \label{eq:moment}
\end{align}
Combining the two inequalities, we obtain
\begin{align*}
\EE [Z^{\ell}] &\le 1 + \ell \, \EE[Z-1] + R \sum_{k=2}^{\floor{\ell}} (\ell \delta)^k + R (\ell \delta)^{\floor{\ell}}(\floor{\ell}+1)\delta\\
  &\le 1 + \ell \, \EE[Z-1] + 2 R (\ell \delta)^2.
\end{align*}
\end{proof}

Our second lemma bounds the upper tail of a certain martingale-like product and is based on
a Bernstein-type inequality. We then derive as a corollary a similar bound on the lower tail.

\begin{lemma}\label{lem:martingaleupper}
Let $R,\delta > 0$ and let $Z_1,\ldots,Z_k$ be non-negative random variables where $k \le 1/(320R\delta^2)$. Assume that for all $1 \le i \le k$,
when conditioning on any values of $Z_1,\ldots,Z_{i-1}$, we almost surely have
\begin{align}
\EE[Z_i ~|~ Z_1, \ldots, Z_{i-1}] &\le 1+ \frac{1}{20k}, \label{eq:matingaleupperexp} \\
\PP[|Z_i - 1| \ge t ~|~ Z_1,\ldots,Z_{i-1}] &\le R \exp(-t/\delta) \quad \text{for all} \ t \geq 0. \label{eq:matingaleupperpr}
\end{align}
Then,
$$
\PP \left[ \prod_{i=1}^k Z_i \ge 1.1 \right] \le \left\{%
\begin{array}{ll}
    \exp(-1/(80\delta)+Rk/2), & k < 1/(80R\delta) \\
    \exp(-1/(12800 R k \delta^2)), & \hbox{otherwise}. \\
\end{array}%
\right.
$$
\end{lemma}
\begin{proof}
Let $2 \leq \ell \leq (2\delta)^{-1}$ be a real number to be determined later on. Then, by Lemma~\ref{lem:moments},
\begin{align*}
\EE\left[ \Big(\prod_{i=1}^k Z_i\Big)^\ell \right] &= \EE_{Z_1,\ldots,Z_{k-1}}\left[ \Big(\prod_{i=1}^{k-1} Z_i\Big)^\ell \EE[Z_k^\ell ~|~ Z_1,\ldots,Z_{k-1}]\right] \\
  &\le \left(1+ \frac{\ell}{20k} + 2R (\ell \delta)^2\right) \EE_{Z_1,\ldots,Z_{k-1}}\left[ \Big(\prod_{i=1}^{k-1} Z_i\Big)^\ell\right] \\
  & \le \cdots \le \left(1+ \frac{\ell}{20k} + 2R (\ell \delta)^2\right)^k \le \exp\left(\frac{\ell}{20} + 2R k (\ell \delta)^2 \right).
\end{align*}
Therefore,
$$ \PP \left[ \prod_{i=1}^k Z_i \ge 1.1 \right] \le 1.1^{-\ell} \exp\left(\frac{\ell}{20} + 2R k (\ell \delta)^2 \right)
\le \exp\left(-\frac{\ell}{40} + 2R k (\ell \delta)^2 \right).$$
The minimum of the right hand side over $\ell$ is $\exp(-1/(12800 R k \delta^2))$ and is obtained
for $\ell = 1/(160 R k \delta^2)$. We set $\ell$ to this value, unless it is greater than $1/(2\delta)$, in which
case we set $\ell = 1/(2\delta)$. The lemma follows.
\end{proof}

\begin{corollary}\label{cor:martingalelower}
Let $R,\delta>0$ and let $Z_1,\ldots,Z_k$ be random variables taking values in $(1/2,\infty)$ where $k \le 1 / (1280 R \delta^2)$.
Assume that for all $1 \le i \le k$, conditioning on any values of $Z_1,\ldots,Z_{i-1}$, we almost surely have
\begin{align*}
\EE[Z_i ~|~ Z_1, \ldots, Z_{i-1}] &\ge 1 - \frac{1}{40k}, \\
\PP[|Z_i - 1| \ge t ~|~ Z_1,\ldots,Z_{i-1}] &\le R \exp(-t/\delta) \quad \text{for all} \ t \ge 0.
\end{align*}
Then,
$$
\PP \left[ \prod_{i=1}^k Z_i \le 0.9 \right] \le \left\{%
\begin{array}{ll}
    \exp(-1/(160 \delta)+Rk/2), & k < 1/(160R\delta) \\
    \exp(-1/(51200 R k \delta^2)), & \hbox{otherwise}. \\
\end{array}%
\right.
$$
\end{corollary}
\begin{proof}
We simply apply Lemma~\ref{lem:martingaleupper} to the random variables $Z_1^{-1},\ldots,Z_k^{-1}$
with $R$ and $2\delta$. Eq.~\eqref{eq:matingaleupperpr} holds because for all $t \ge 0$ and $x \ge 1/2$
if $|x^{-1}-1|\ge t$ then also $|x-1|\ge t/2$. For Eq.~\eqref{eq:matingaleupperexp}, we use the inequality
$x^{-1} \le 1-(x-1)+2(x-1)^2$, valid for all $x \ge 1/2$. This implies that
\begin{align*}
\EE[Z_i^{-1} ~|~ Z_1, \ldots, Z_{i-1}] &\le 1 + \frac{1}{40k} + 2 \EE[(Z_i-1)^{2} ~|~ Z_1, \ldots, Z_{i-1}] \\
&\le 1 + \frac{1}{40k} + 4 R \delta^2 \le 1 + \frac{1}{20k},
\end{align*}
where the next-to-last inequality follows from the calculation in~\eqref{eq:moment}.
\end{proof}

\begin{proof}[Proof of Theorem~\ref{thm:mainsamplinglowdim}]
Fix $1 \leq m \leq 9n / 10$ and a set $A \subseteq S^{n-1}$.
Consider a sequence of random subspaces in $\RR^n$,
$$ \RR^n = H_0 \supset H_1 \supset H_2 \supset \cdots \supset H_m $$
in which $\dim(H_i) = n-i$, defined as follows. For each $i \geq 1$, the subspace $H_{i}$ is chosen uniformly in the Grassmannian of all $(n-i)$-dimensional
subspaces of $H_{i-1}$. An important observation, which follows from the uniqueness of the Haar measure,
is that the subspace $H_i$ is distributed uniformly over $\cG_{n, n-i}$, and in particular,
$H_m$ is a uniform $(n-m)$-dimensional subspace.

For $k=1,\ldots,m$ define the random variable
$$ X_k = \frac{\sigma_{H_k}(A \cap H_k)}{\sigma_{H_{k-1}}(A \cap H_{k-1})},$$
where $\sigma_{H_k}$ is the uniform measure on the sphere $S^{n-1} \cap H_k$.
If the denominator vanishes, we set the random variable to $1$.
Notice that
$$ \prod_{k=1}^m X_k = \frac{\sigma_{H_m}(A \cap H_m)}{\sigma(A)} $$
and hence our goal is to show that this product is in $1 \pm 0.1$ except with probability at most $C \exp(-c n^{1/3})$.
We will do this by applying Lemma~\ref{lem:martingaleupper} and Corollary~\ref{cor:martingalelower} to
a regularized version of $X_1,\ldots,X_m$ defined below.

We note three properties of the random variables $X_k$. First, we have that for any $1 \leq k \leq m$, conditioned on any values of $H_1,\ldots,H_{k-1}$,
$$
 \EE \left( X_k | H_1,\ldots,H_{k-1} \right) = 1.
$$
This holds since $H_k$ is distributed uniformly over the Grassmannian of subspaces of $H_{k-1}$.
Second, by definition, $X_k$ is bounded from above by $1/(\sigma_{H_{k-1}}(A \cap H_{k-1}))$.
Finally, by Theorem~\ref{thm:mainequator}, for any $0<t<1$,
\begin{align*}
 \PP(|X_k-1| \ge t ~|~ H_1,\ldots,H_{k-1}) & \le C \exp(-c (n-k+1) t / \log (2/\sigma_{H_{k-1}}(A \cap H_{k-1}))) \\
&\le C \exp(-\tilde{c} n t / \log (2/\sigma_{H_{k-1}}(A \cap H_{k-1}))),
\end{align*}
where we used the fact that $k \le m \le 9n/10$.
Because this tail bound holds only for $t<1$, we cannot apply
Lemma~\ref{lem:martingaleupper} and Corollary~\ref{cor:martingalelower} directly,
and instead proceed below to define the regularized random variables $Z_1,\ldots,Z_m$.

Next, for $0 \le k \le m$, we define the ``bad" event $B_k$ as the event that $X_1 X_2 \cdots X_k \le 1/2$
and for $1 \le k \le m$, the ``bad" event $C_k$ as the event that $|X_k-1| \ge 1/2$.
Condition on any $H_1,\ldots,H_{k-1}$ such that $B_{k-1}$
does not occur. In this case, $\sigma_{H_{k-1}}(A \cap H_{k-1}) \ge \sigma(A)/2$. Hence,
$X_k$ is upper bounded by $2/\sigma(A) \le C \exp(c n^{1/3})$. Moreover, for any $0<t<1$ the probability that $|X_k-1| \ge t$ is at most
$C \exp(-c n t / \log (4/ \sigma(A))) \le C \exp(-\tilde{c} n^{2/3} t)$, and in particular the probability
that $C_k$ occurs is at most $C \exp(-c n^{2/3})$.
For $1 \le k \le m$, we define the random variable $Z_k$ as follows: if either $B_{k-1}$
or $C_k$ occurs, $Z_k$ is $1$. Otherwise, $Z_k = X_k$.

We can now finally apply Lemma~\ref{lem:martingaleupper} and Corollary~\ref{cor:martingalelower}:
for each $1 \le k \le m$, we apply them to the sequence $Z_1,\ldots,Z_k$ with $R=C$ and $\delta = \tilde{C} n^{-2/3}$.
To see why the conditions there hold, condition on any $H_1,\ldots,H_{k-1}$, and assume first that $B_{k-1}$ does not occur. Then
\begin{align*}
|\EE[Z_k | H_1,\ldots,H_{k-1}]-1| &= |\EE[Z_k -X_k| H_1,\ldots,H_{k-1}]| \\
 &\le \PP[C_k | H_1,\ldots,H_{k-1}] \cdot C \exp(c n^{1/3}) \\
 &\le \tilde{C} \exp(-\tilde{c} n^{2/3}).
\end{align*}
Moreover, for \emph{all} non-negative $t$, the probability that $|Z_k-1| \ge t$ is at most
$C \exp(-c n^{2/3} t)$. Finally, these two statements are obviously true even when $B_{k-1}$ does
occur (since in this case $Z_k$ is simply $1$), hence we obtain that the two statements hold
conditioned on any $H_1,\ldots,H_{k-1}$ (and in particular, on any $Z_1,\ldots,Z_{k-1}$).
As a result, the lemma and the corollary imply that for each $1 \le k \le m$,
$|Z_1\cdots Z_k - 1| \ge 0.1$ with probability at most $C \exp(-c n^{1/3})$.
Moreover, by a union bound, the probability that there exists a $k$ for which $|Z_1\cdots Z_k - 1| \ge 0.1$,
an event which we denote by $D$, is at most
\begin{align}
 \PP[D] \le m \cdot C \exp(-c n^{1/3}) \le \tilde{C} \exp(-\tilde{c} n^{1/3}). \label{eq:boundonbadaeventd}
\end{align}

Next, we claim that for any $1 \le k \le m$,
\begin{align}
\PP[\neg D \wedge \neg C_1 \wedge \cdots \wedge \neg C_{k-1} \wedge C_k] \le C \exp(-c n^{2/3}). \label{eq:boundonbadaeventdcc}
\end{align}
To see why, notice that $\neg C_1$ implies that $X_1=Z_1$, which together with $\neg D$ implies
that $\neg B_1$; the latter, in turn,
implies that $X_2=Z_2$ (since neither $B_1$ nor $C_2$ happens), which implies that $B_2$ does
not happen either; etc. As a result, we obtain that $\neg B_{k-1}$, which implies
that the probability of $C_k$ is at most $C \exp(-c n^{2/3})$, as desired.

By summing all the probabilities in~\eqref{eq:boundonbadaeventd} and~\eqref{eq:boundonbadaeventdcc}, we obtain
that
$$
\PP[\neg D \wedge \neg C_1 \wedge \cdots \wedge \neg C_m] \ge 1 - C \exp(-c n^{1/3}).
$$
It remains to notice using the same argument as above that this event implies that for all $k$, $Z_k=X_k$ and therefore also that
$|X_1\cdots X_m - 1| < 0.1$.
\end{proof}

The only thing remaining is to derive Lemma~\ref{lem:main2}
from Theorem~\ref{thm:mainsamplinglowdim}. We restate it here in a slightly more
general form.

\begin{lemma}
Let $1 \leq m \leq 9n / 10$.
Suppose that $A \subseteq S^{n-1}$ and $B \subseteq \cG_{n,n-m}$
are measurable sets with
$$ \sigma(A) \geq C \exp(-c n^{1/3}), \qquad \sigma_{\cG}(B) \geq C \exp(-c n^{1/3}) $$
for some universal constants $c,C>0$. Then,
$$ \sigma_{\cI} \left( (A \times B) \cap \cI_{n,m} \right) \geq 0.8\, \sigma(A) \sigma_{\cG}(B).$$
\end{lemma}
\begin{proof}
Notice that $ \sigma_{\cI} ( (A \times B) \cap \cI_{n,m} ) / \sigma_{\cG}(B)$ may
be interpreted as the probability that when choosing a subspace $H$ uniformly from $B$
and a uniform vector $x$ in $H \cap S^{n-1}$, we have $x \in A$. To analyze this probability,
denote by $E \subseteq \cG_{n,n-m}$
the set of all $(n-m)$-dimensional subspaces $H$ for which
$$ \frac{\sigma_H(A \cap H)}{\sigma(A)} \le 0.9.$$
Then, by Theorem~\ref{thm:mainsamplinglowdim}, $\sigma_{\cG}(E) \leq C \exp (-c n^{1/3})$.
Next, observe that the probability that $H \in E$ is at most $\sigma_{\cG}(E) / \sigma_{\cG}(B)$.
Moreover, if $H \notin E$, then by definition, the probability that $x \in A$ is at least $0.9 \, \sigma(A)$.
Hence,
$$ \frac{\sigma_{\cI} ( (A \times B) \cap \cI_{n,m} )}{\sigma_{\cG}(B)} \ge
   \left( 1 - \frac{\sigma_{\cG}(E)}{\sigma_{\cG}(B)} \right) 0.9 \, \sigma(A) > 0.8 \, \sigma(A),$$
assuming the universal constants are chosen properly.
\end{proof}

\subsection*{Acknowledgments}

We thank Ronald de Wolf for comments on an earlier draft.

\bibliographystyle{alphaabbrvprelim}
\bibliography{communication}

\end{document}